\documentclass[runningheads,a4paper]{llncs}
\usepackage{amsmath}
\usepackage{amssymb}
\usepackage{graphicx}
\usepackage{mathrsfs}
\usepackage{graphicx}
\usepackage{animate}
\usepackage{algorithm}
\usepackage{algorithmicx}
\usepackage{algpseudocode}

\newtheorem{assumption}{Assumption}

\usepackage{url}
\urldef{\mailsa}\path|{weiming.xiang, diego.manzanas.lopez, patrick.musau, taylor.johnson}@vanderbilt.edu|     
\newcommand{\keywords}[1]{\par\addvspace\baselineskip
\noindent\keywordname\enspace\ignorespaces#1}

\begin{document}

\mainmatter  

\title{Reachable Set Estimation and Verification for Neural Network Models of Nonlinear Dynamic Systems}

\titlerunning{Reachable Set Estimation and Verification for Neural Network Model}

%
%
\author{Weiming Xiang
\and Diego Manzanas Lopez \and Patrick Musau \and Taylor T. Johnson}
\authorrunning{W. Xiang \and  D. Manzanas \and  P. Musau \and  T. T. Johnson}

\institute{Department of Electrical Engineering and Computer Science,
\\Vanderbilt University,  Nashville, Tennessee 37212, USA\\
\mailsa}

%
%

\toctitle{Lecture Notes in Computer Science}
\tocauthor{Authors' Instructions}
\maketitle

\begin{abstract}

Neural networks have been widely used to solve complex real-world problems. Due to the complicate, nonlinear, non-convex nature of neural networks, formal safety guarantees for the behaviors of neural network systems will be crucial for their applications in safety-critical systems. In this paper, the reachable set estimation and verification problems for Nonlinear Autoregressive-Moving Average (NARMA) models in the forms of neural networks are addressed. The neural network involved in the model is a class of feed-forward neural networks called Multi-Layer Perceptron (MLP). By partitioning the input set of an MLP into a finite number of cells, a layer-by-layer computation algorithm is developed for reachable set estimation for each individual cell. The union of estimated reachable sets of all cells forms an over-approximation of reachable set of the MLP. Furthermore, an iterative reachable set estimation algorithm based on reachable set estimation for MLPs is developed for NARMA models. The safety verification can be performed by checking the existence of intersections of unsafe regions and estimated reachable set. Several numerical examples are provided to illustrate our approach.

\keywords{Neural network; Reachable set estimation; Safety verification; Nonlinear systems}
\end{abstract}

\section{Introduction}
Artificial neural networks have been widely used in machine learning systems. Applications include adaptive control \cite{hunt1992neural,ge1999adaptive}, pattern recognition \cite{schmidhuber2015deep,lawrence1997face}, game playing \cite{silver2016mastering}, autonomous vehicles \cite{bojarski2016end}, and many others. Neural networks are trained over a finite number of input and output data, and are expected to be able to generalize to produce desirable outputs for given inputs even including previously unseen inputs. Though neural networks have been showing effectiveness and powerful ability in resolving complex problems, they are confined to systems which comply only to the lowest safety integrity levels since, most of the time, a neural network is viewed as a \emph{black box} without effective methods to assure safety specifications for its outputs.  For nonlinear dynamic systems whose models are difficult or even impossible to establish, using neural network models that are inherently derived from input and output data to approximate the nonlinear dynamics is an efficient and practical way. One standard employment of neural networks is to approximate the Nonlinear Autoregressive-Moving Average (NARMA) model which is a popular model for nonlinear dynamic systems. However, once the NARMA model in the form of neural networks is established, a problem naturally arises: \emph{How to compute the reachable set of an NARMA model that is essentially expressed by neural networks and, based on that, how to verify properties of an NARMA model?} For computing or estimating the reachable set for a nonlinear system starting from an initial set and with an input set, the numbers of inputs and initial state that need to be checked are infinite, which is impossible only by performing experiments. Moreover, it has been observed that neural networks can react in unexpected and incorrect ways to even slight perturbations of their inputs \cite{szegedy2013intriguing}, which could result in unsafe systems. Hence, methods that are able to provide formal guarantees are in a great demand for verifying specifications or properties of systems involving neural networks.  Verifying neural networks is a hard problem, even simple properties about them have been proven NP-complete
problems \cite{katz2017reluplex}. The difficulties mainly come from the presence of activation functions and the complex structures, making neural networks large-scale, nonlinear, non-convex and thus incomprehensible to humans. Until now, only few results have been reported for verifying neural networks. The verification for feed-forward
multi-layer neural networks is investigated based on \emph{Satisfiability Modulo Theory} (SMT) in \cite{huang2016safety,pulina2012challenging}. In \cite{pulina2010abstraction} an abstraction-refinement approach is proposed for verification of
specific networks known as \emph{Multi-Layer Perceptrons} (MLPs). In \cite{xiang2017reachable_arxiv,katz2017reluplex}, a specific kind of activation functions called \emph{Rectified Linear Unit} (ReLU) is considered for the verification problem of neural networks. A simulation-based approach is developed in \cite{xiang2017output_arxiv}, which turns the reachable set estimation problem into a neural network maximal sensitivity computation problem that is described in terms of a chain of convex optimization problems. Additionally, some recent reachable set/state estimation results are reported for neural networks \cite{zhang2017state,zhang2016synchronization,xu2017reachable,zuo2014non,thuan2016reachable}, these results that are based on Lyapunov functions analogous to stability \cite {xiang2018Parameter,xiang2017robust,xiang2017stability,zhang2016mode,xiang2014stabilization,xiang2015equivalence,xiang2016necessary} and reachability analysis of dynamical systems \cite{xiang2017output,xiang2017reachable}, certainly have potentials to be further extended to safety verification.

In this paper, we will use neural networks to represent the forward dynamics of the nonlinear systems that are in the form of NARMA models. Due to the non-convex and nonlinearity existing in the model and inspired by some simulation-based ideas for verification problems \cite{duggirala2015c2e2,fan2016automatic,bak2017hylaa,bak2017rigorous},  a simulation-based approach will be developed to estimate the reachable set of state responses generated from a NARMA model. The core step of the approach is the reachable set estimation for a class of feed-forward neural networks called Multi-Layer Perceptron (MLP). By discretizing the input space of an MLP into a finite-number of regularized cells, a layer-by-layer computation process is developed to establish an over-approximation of the output set for each individual cell. The union of output set of all cells is the reachable set estimation for the MLP with respect to a given input set. On the basis of the reachable set estimation method for MLPs, the reachable set over any finite-time interval for an NARMA model can be estimated in a recursive manner. Safety verification can be performed if an estimation for the reachable set of an NARMA model is established, by checking the existence of intersections between the estimated reachable set and unsafe regions. 

The remainder of this paper is organized as follows. Neural network model of nonlinear systems, that is the NARMA model, is introduced in Section 2. The problem formulation is presented in Section 3. The main results, reachable set estimation for MLPs and NARMA models, are given in Sections 4 and 5, respectively. Conclusions are made in Section 6. 

\emph{Notations:}  $\mathbb{R}$ denotes the field of real numbers, $\mathbb{R}^n$  stands for the vector space of all $n$-tuples of real numbers, $\mathbb{R}^{n\times n}$  is the space of $n\times n$  matrices with real entries.  $\left\|  \mathbf{x}  \right\|_\infty$  stands for infinity norm for vector $\mathbf{x} \in \mathbb{R}^{n}$ defined as $\left\|  \mathbf{x}  \right\|_\infty = \max\nolimits_{i=1,\ldots,n}{\left|x_i\right|}$.  $\mathbf{A}^{\top}$ denotes the transpose of matrix $\mathbf{A}$. For a set $\mathcal{A}$, $\left|\mathcal{A}\right|$ denotes its cardinality.

\section{Neural Network Models of Nonlinear Dynamic Systems}

Neural networks are commonly used for data-driven modeling for nonlinear systems. 
One standard model to represent discrete-time nonlinear systems is the Nonlinear Autoregressive-Moving Average (NARMA) model. Given a discrete-time process with past
states $\mathbf{x}(k), \mathbf{x}(k - 1), \ldots, \mathbf{x}(k-d_x)$ and inputs $\mathbf{u}(k), \mathbf{u}(k- 1), \ldots, \mathbf{u}(k-d_u)$, an NARMA model is in the  form of 
\begin{equation}\label{NARMA}
\mathbf{x}(k+1)=f\left(\mathbf{x}(k),\mathbf{x}(k-1),\ldots,\mathbf{x}(k-d_x),\mathbf{u}(k),\mathbf{u}(k-1),\ldots, \mathbf{u}(k-d_u)\right), \end{equation}
where the nonlinear function $f(\cdot)$ needs to be approximated by training neural networks. The initial state of NARMA model (\ref{NARMA}) is $\mathbf{x}(0),\ldots,\mathbf{x}(d_x)$, which is assumed to be  in  set  $\mathcal{X}_0 \times\cdots\times\mathcal{X}_{d_x}$, 
and the input set is $\mathcal{U}$. We assume that the initial state $\{\mathbf{x}(0),\ldots,\mathbf{x}(d_x)\} \in \mathcal{X}_0 \times\cdots\times\mathcal{X}_{d_x}$ and input satisfies  $\mathbf{u}(k)\in \mathcal{U}$, $\forall k \in \mathbb{N}$. 


A neural network consists of a number of interconnected neurons. Each neuron is a simple processing element that responds to the weighted inputs it received from other neurons. In this paper, we consider the most popular and general feed-forward neural network, MLP. Generally, an MLP consists of three typical classes of layers: An input layer, that
serves to pass the input vector to the network, hidden layers of computation neurons, and
an output layer composed of at least a computation neuron to produce the output vector.

The action of a neuron depends on its activation function, which is described as 
\begin{align}
y_i = h\left(\sum\nolimits_{j=1}^{n}\omega_{ij} v_j + \theta_i\right),
\end{align}
where $v_j$ is the $j$th input of the $i$th neuron, $\omega_{ij}$ is the weight from the $j$th input to the $i$th neuron, $\theta_i$ is called the bias of the $i$th neuron, $y_i$ is the output of the $i$th neuron, $h(\cdot)$ is the activation function. The activation function is generally a nonlinear function  describing the reaction of $i$th neuron with inputs $v_j$, $j=1,\cdots,n$. Typical activation functions include Rectified Linear Unit (ReLU), logistic, tanh, exponential linear unit, linear functions, etc. In this work, our approach aims at dealing with  activation functions regardless of their specific forms, only the following monotonic assumption needs to be satisfied.

\begin{assumption}\label{assumption_1}
	For any $v_1 \le v_2$, the activation function satisfies $h(v_1) \le h(v_2)$. 
\end{assumption}
	Assumption \ref{assumption_1} is a common property that can be satisfied by a variety of activation functions. For example, it is easy to verify that the most commonly used such as logistic, tanh,  ReLU, all satisfy Assumption \ref{assumption_1}.

An MLP has multiple layers,  each layer $\ell$, $1 \le \ell \le L $, has $n^{[\ell]}$ neurons.  In particular, layer $\ell =0$ is used to denote the input layer and $n^{[0]}$ stands for the number of inputs in the rest of this paper, and of course, $n^{[L]}$ stands for the last layer, that is the output layer. For a neuron $i$, $1 \le i \le n^{[\ell]}$ in layer $\ell$, the corresponding input vector is denoted by $\mathbf{v}^{[\ell]}$ and the weight matrix is 
\begin{equation}
\mathbf{W}^{[\ell]} = \left[\boldsymbol{\omega}_{1}^{[\ell]},\ldots,\boldsymbol{\omega}_{n^{[\ell]}}^{[\ell]}\right]^{\top},
\end{equation}
where $\boldsymbol{\omega}_{i}^{[\ell]}$ is the weight vector. The bias vector for layer $\ell$ is
\begin{equation*} \boldsymbol{\theta}^{[\ell]}=\left[\theta_1^{[\ell]},\ldots,\theta_{n^{[\ell]}}^{[\ell]}\right]^{\top}
\end{equation*} 

The output vector of layer $\ell$ can be expressed as 
\begin{equation}
\mathbf{y}^{[\ell]}=h_{\ell}(\mathbf{W}^{[\ell]}\mathbf{v}^{[\ell]}+\boldsymbol{\theta}^{[\ell]}),
\end{equation} 
where $h_{\ell}(\cdot)$ is the activation function for layer $\ell$.

For an MLP, the output of $\ell-1$ layer is the input of $\ell$ layer, and the mapping from the input of input layer $\mathbf{v}^{[0]}$ to the output of output layer $\mathbf{y}^{[L]}$ stands for the input-output relation of the MLP, denoted by
\begin{equation}\label{NN}
\mathbf{y}^{[L]} = H (\mathbf{v}^{[0]}),
\end{equation}    
where $H(\cdot) \triangleq h_L  \circ h_{L - 1}  \circ  \cdots  \circ h_1(\cdot) $.

According to the \emph{Universal Approximation Theorem} \cite{hornik1989multilayer}, it guarantees that, in
principle, such an MLP in (\ref{NN}), namely the function $F(\cdot)$, is able to approximate any nonlinear real-valued function. To use MLP (\ref{NN}) to approximate NARMA model (\ref{NARMA}), we can let the input of (\ref{NN}) as
\begin{equation} \label{NN_v0}
\mathbf{v}^{[0]} = [\mathbf{x}^{\top}(k),\mathbf{x}^{\top}(k-1),\ldots,\mathbf{x}^{\top}(k-d_x),\mathbf{u}^{\top}(k),\mathbf{u}^{\top}(k-1),\ldots,\mathbf{u}^{\top}(k-d_u)]^{\top},
\end{equation} 
and output as
\begin{equation}  \label{NN_yL}
\mathbf{y}^{[L]} = \mathbf{x}(k+1).
\end{equation}

With the input and output data of original nonlinear systems, an approximation of NARMA model (\ref{NARMA}) can be obtained by standard feed-forward neural network training process. Despite the impressive ability of approximating nonlinear functions,  much complexities represent in predicting the output behaviors of MLP  (\ref{NN}) as well as NARMA model (\ref{NARMA}) because of the nonlinearity and non-convexity of MLPs. In the most of real applications, an MLP is usually viewed as a \emph{black box} to generate  a desirable output with respect to a given input. However, regarding  property verifications such as the safety verification, it has been observed that even a well-trained neural network can react in unexpected and incorrect ways to even slight perturbations of their inputs, which could result in unsafe systems. Thus,  to validate the neural network NARMA model for a nonlinear dynamics, it is necessary to compute the reachable set estimation of the model, which is able to cover all possible values of output, to assure that the state trajectories of the model will not attain unreasonable values that is inadmissible for the original system. It is also necessary to estimate all possible values of state for safety verification of a neural network NARMA model. 

\section{Problem Formulation}

Consider initial set $\mathcal{X}_0 \times\cdots\times\mathcal{X}_{d_x}$ and input set $\mathcal{U}$, the reachable set of NARMA model in the form of (\ref{NARMA}) is defined as follows.
\begin{definition} \label{reachable_set}
	Given an NARMA model in the form of (\ref{NARMA}) with initial set  $\mathcal{X}_0 \times\cdots\times\mathcal{X}_{d_x}$ and input set $\mathcal{U}$, the reachable set at a time instant $k$ is:	\begin{align}
	\mathcal{X}_k \triangleq \{\mathbf{x}(k)  \mid \mathbf{x}(k)~\mathrm{satisfies}~ (\ref{NARMA})~\mathrm{and}~\{\mathbf{x}(0),\cdots,\mathbf{x}(d_x)\} \in \mathcal{X}_0 \times\ldots\times\mathcal{X}_{d_x}, \nonumber
	\\
	\mathbf{u}(k) \in \mathcal{U},~\forall k \in \mathbb{N}\}, \label{def_1}
	\end{align}
	and the reachable set over time interval $[0,k_f]$ is defined by
	\begin{equation}\label{def_2}
	\mathcal{X}_{[0,k_f]} = \bigcup\nolimits_{s=0}^{k_f}\mathcal{X}_{s}.
	\end{equation} 
\end{definition}

Since MLPs are often large, nonlinear, and non-convex, it is extremely difficult to compute the exact reachable set $\mathcal{X}_k$ and $\mathcal{X}_{[0,k_f]}$ for an NARMA model with MLPs. Rather than directly computing the exact output reachable set for an  NARMA model, a more practical and feasible way is to derive an over-approximation of $\mathcal{X}_k$, which is called  reachable set estimation.

\begin{definition}\label{reachable_estimation}
	A set  $\tilde{\mathcal{X}_k} $ is called a  reachable set estimation of NARMA model (\ref{NARMA})  at time instant $k$, if $\mathcal{X}_k \subseteq \tilde{\mathcal{X}_k}$ holds and, moreover, $\tilde{\mathcal{X}}_{[0,k_f]} = \bigcup\nolimits_{s=0}^{k}\tilde{\mathcal{X}}_{s}$ is a reachable set estimation for NARMA model (\ref{NARMA})  over time interval $[0,k_f]$. 
\end{definition}

Based on Definition \ref{reachable_estimation}, the problem of reachable set estimation for an NARMA model is given as below.

\begin{problem} \label{problem1}
	How does one find the set $\tilde{\mathcal{X}}_k $ such that $\mathcal{X}_k \subseteq \tilde{\mathcal{X}_k}$, given a bounded initial set  $\mathcal{X}_0 \times\ldots\times\mathcal{X}_{d_x}$ and an input set $\mathcal{U}$ and an NARMA model (\ref{NARMA})? 
\end{problem}

In this work, we will focus on the safety verification for NARMA models. The safety specification for output is expressed by a set defined in the state space, describing the safety requirement. 

\begin{definition}
	Safety specification $\mathcal{S}$ 
	formalizes the safety requirements for  state $\mathbf{x}(k)$ of  NARMA model (\ref{NARMA}), and is a predicate over  state $\mathbf{x}$ of NARMA model (\ref{NARMA}). The NARMA model (\ref{NARMA}) is safe over time interval $[0,k_f]$ if and only if the following condition is satisfied:
	\begin{equation}\label{safety}
	\mathcal{X}_{[0,k_f]} \cap \neg \mathcal{S} = \emptyset,
	\end{equation}
	where $\neg $ is the symbol for logical negation. 
\end{definition}

Therefore, the safety verification problem for NARMA models is stated as follows.
\begin{problem}\label{problem2}
	How can the safety requirement in (\ref{safety}) be verified given an NARMA model (\ref{NARMA}) with a bounded initial set  $\mathcal{X}_0 \times\ldots\times\mathcal{X}_{d_x}$ and an input set $\mathcal{U}$ and a safety specification $\mathcal{S}$? 
\end{problem}

Before ending this section, a lemma is presented to show that the safety verification of an MLP can be relaxed by checking with the over-approximation of output reachable set. 

\begin{lemma}\label{lemma1}
	Consider an NARMA model (\ref{NARMA}) and a safety specification $\mathcal{S}$, the NARMA model is safe in time interval $[0,k_f]$ if the following condition is satisfied
	\begin{equation}\label{estimate_safety}
	\tilde{\mathcal{X}}_{[0,k_f]} \cap \neg \mathcal{S} = \emptyset,
	\end{equation}
	where $\mathcal{X}_{[0,k_f]}\subseteq \tilde{\mathcal{X}}_{[0,k_f]}$.
\end{lemma}
\begin{proof}
	Since $\mathcal{X}_{[0,k_f]}\subseteq \tilde{\mathcal{X}}_{[0,k_f]}$, condition (\ref{estimate_safety}) directly leads to $\mathcal{X}_{[0,k_f]} \cap \neg \mathcal{S} = \emptyset$. The proof is complete.
\end{proof}

Lemma \ref{lemma1} implies that it is sufficient to use the estimated reachable set for the safety verification of  an NARMA model, thus the solution of Problem \ref{problem1} is also the key to solve Problem \ref{problem2}. 

\section{Reachable Set Estimation for MLPs}
As (\ref{NN})--(\ref{NN_yL}) in previous section, the state of an NARMA model is computed through an MLP recursively. Therefore, the first step for the reachable set estimation for an NARMA model is to estimate the output set of MLP (\ref{NN}). 

Given an MLP $\mathbf{y}^{[L]} = H(\mathbf{v}^{[0]})$ with a bounded input set $\mathcal{V}$, the problem is how to compute a set $\mathcal{Y}$ as below:
\begin{equation}
\mathcal{Y} \triangleq \{\mathbf{y}^{[L]} \mid \mathbf{y}^{[L]} = H(\mathbf{v}^{[0]}),~\mathbf{v}^{[0]} \in \mathcal{V} \subset \mathbb{R}^{n}\}.
\end{equation}

Due to the complex structure and nonlinearities in activation functions, the estimation of output reachable set of MLP represents much difficulties if only using analytical methods.
One possible way to circumvent those difficulties is to employ the information produced by a finite number of simulations.

\begin{definition}
	Given a set $\mathcal{V} \subset \mathbb{R}^{n}$, a finite collection of sets $\mathscr{P} = \{\mathcal{P}_1,\mathcal{P}_2,\ldots,\mathcal{P}_N\}$ is said to be a partition of $\mathcal{V}$ if (1) $\mathcal{P}_i \subseteq \mathcal{V}$; (2) $\mathrm{int}(\mathcal{P}_i) \cup \mathrm{int}(\mathcal{P}_j) = \emptyset$; (3) $\mathcal{V} \subseteq \bigcup\nolimits_{i=1}^{N}\mathcal{P}_i$, $\forall i \in \{1,\ldots,N\}$. Each elements $\mathcal{P}_i$ of partition $\mathscr{P}$ is called a cell.
\end{definition}

In this paper, we use cells defined by intervals which are given as follows: For any bounded set $\mathcal{V} \subset \mathbb{R}^{n}$, we have $\mathcal{V} \subseteq \bar{\mathcal{V}}$,where $\bar{\mathcal{V}} = \{\mathbf{v} \in \mathbb{R}^{n} \mid \underline{\mathbf{v}} \le \mathbf{v} \le \bar{\mathbf{v}}\}$, in which $\underline{\mathbf{v}}$ and $\bar{\mathbf{v}}$ are defined as the lower and upper bounds of elements of $\mathbf{v}$ in $\mathcal{V}$ as  $\underline{\mathbf{v}}=[\mathrm{inf}_{\mathbf{v}\in \mathcal{V}}(v_1),\ldots,\mathrm{inf}_{\mathbf{v}\in \mathcal{V}}(v_n)]^{\top}$ and $\bar{\mathbf{v}}=[\mathrm{sup}_{\mathbf{v}\in \mathcal{V}}(v_1),\ldots,\mathrm{sup}_{\mathbf{v}\in \mathcal{V}}(v_n)]^{\top}$, respectively. Then, we are able to partition interval $\mathcal{I}_i=[\mathrm{inf}_{\mathbf{v}\in \mathcal{V}}(v_i),~\mathrm{sup}_{\mathbf{v}\in \mathcal{V}}(v_i)]$, $i \in \{1,\ldots,n\}$ into $M_i$ segments as $\mathcal{I}_{i,1}=[v_{i,0},v_{i,1}]$, $\mathcal{I}_{i,2}=[v_{i,1},v_{i,2}]$, $\ldots$, $\mathcal{I}_{i,M_i}=[v_{i,M_i-1},v_{i,M_i}]$, where $v_{i,0}=\mathrm{inf}_{\mathbf{v}\in \mathcal{V}}(v_i)$, $v_{i,M_i}=\mathrm{sup}_{\mathbf{v}\in \mathcal{V}}(v_i)$ and $v_{i,n} = v_{i,0} + \frac{m(v_{i,M_i}-v_{i,0})}{M_i}$, $m \in \{0,1,\ldots,M_i\}$. The cells then can be constructed as $\mathcal{P}_i = \mathcal{I}_{1,m_1} \times\cdots\times \mathcal{I}_{n,m_n}$, $i \in \{1,2,\ldots,\prod\nolimits_{s=1}^{n}M_s\}$, $\{m_1,\ldots,m_n\}\in \{1,\ldots,M_1\} \times \cdots \times \{1,\ldots,M_n\}$. To remove redundant cells, we have to check if the cell has empty intersection with $\mathcal{V}$. Cell $\mathcal{P}_i$ should be removed if $\mathcal{P}_i \cap \mathcal{V} = \emptyset$.  The cell construction process is summarized by \texttt{cell} function in Algorithm \ref{alg1}.

\begin{algorithm}\caption{Partition an input set} \label{alg1}
\begin{algorithmic}[1]
	\Require Set $\mathcal{V}$, partition numbers $M_i$, $i \in \{ 1,\ldots,n\}$
	\Ensure Partition $\mathscr{P} = \{\mathcal{P}_1,\mathcal{P}_2,\ldots,\mathcal{P}_N\}$
	
	\Function{cell}{$\mathcal{V}$, $M_i$, $i \in \{ 1,\ldots,n\}$}
	
	\State $v_{i,0} \gets \mathrm{inf}_{\mathbf{v}\in \mathcal{V}}(v_i)$, $v_{i,M_i} \gets \mathrm{sup}_{\mathbf{v}\in \mathcal{V}}(v_i)$
	
	\For{$i = 1:1:n$}
	\For{$j = 1:1:M_i$}
	\State $v_{i,j} \gets v_{i,0} + \frac{j(v_{i,M_i}-v_{i,0})}{M_i}$
	\State $\mathcal{I}_{i,j} \gets [v_{i,j-1},v_{i,j}]$
	\EndFor
	\EndFor
	
	\State $\mathcal{P}_i \gets \mathcal{I}_{1,m_1} \times\cdots\times \mathcal{I}_{n,m_n}$, $\{m_1,\ldots,m_n\}\in \{1,\ldots,M_1\} \times \cdots \times \{1,\ldots,M_n\}$
	\If{$\mathcal{P}_i \cap \mathcal{V} = \emptyset$}
	\State Remove $\mathcal{P}_i$
	\EndIf
	
	\State\Return $\mathscr{P} = \{\mathcal{P}_1,\mathcal{P}_2,\ldots,\mathcal{P}_N\}$
	
	\EndFunction
\end{algorithmic}
\end{algorithm}

With the cells constructed by \texttt{cell} function, the next step is to develop a function that is able to estimate the output reachable set for each individual cell as the input to the MLP. A layer-by-layer approach is developed.

\begin{theorem}\label{thm1}
	For a single layer $\mathbf{y}=h(\mathbf{W}\mathbf{v}+\boldsymbol{\theta})$, if the input set is a cell described by  $\mathcal{I}_{1} \times\cdots\times \mathcal{I}_{n_v}$ where $\mathcal{I}_{i} = [\underline{v}_i,\bar v_i]$, $i \in \{1,\ldots,n_v\}$, the output set can be over-approximated by a cell in the expression of intervals $\mathcal{I}_{1} \times\cdots\times \mathcal{I}_{n_y}$, where $\mathcal{I}_{i}$, $i \in \{1,\ldots,n_y\}$ can be computed by 
	\begin{align}\label{thm1_1}
	\mathcal{I}_{i} = 	[h (\underline{z}_i + \theta_i) ,~h (\bar{z}_i +  \theta_i) ],
	\end{align}
	where $\underline{z_i} = \sum\nolimits_{j=1}^{n_v}\underline{g}_{ij}$, $\bar{z_i} = \sum\nolimits_{j=1}^{n_v}\bar{g}_{ij}$ with $\underline{g}_{ij}$ and $\bar{g}_{ij}$ defined by
	\begin{equation}\label{thm1_2}
\underline{g}_{ij}  = \left\{ {\begin{array}{*{20}c}
	{\omega _{ij} \underline{v}_j } & {\omega _{ij}  \geq 0}  \\
	{\omega _{ij} \bar v_j } & {\omega _{ij}  < 0}  \\
	\end{array} } \right.,~\bar g_{ij}  = \left\{ {\begin{array}{*{20}c}
	{\omega _{ij} \bar v_j } & {\omega _{ij}  \geq 0}  \\
	{\omega _{ij} \underline{v}_j } & {\omega _{ij}  < 0}  \\
	\end{array} } \right..
	\end{equation}
\end{theorem}
\begin{proof}	
	By (\ref{thm1_2}), one can obtain that 
	\begin{align}
	\underline{z}_i = \min_{\mathbf{v} \in \mathcal{I}_{1} \times\cdots\times \mathcal{I}_{n_v}}\left(\sum\nolimits_{j=1}^{n_v}\omega_{ij} v_j\right),
	\\
	\bar{z}_i = \max_{\mathbf{v} \in \mathcal{I}_{1} \times\cdots\times \mathcal{I}_{n_v}}\left(\sum\nolimits_{j=1}^{n_v}\omega_{ij} v_j\right).
	\end{align}

Consider neuron $i$, its output is $y_i =  h\left(\sum\nolimits_{j=1}^{n_v}\omega_{ij}v_j + \theta_i\right)$. Under Assumption \ref{assumption_1}, we can conclude that 
		\begin{align}
		\min_{\mathbf{v} \in \mathcal{I}_{1} \times\cdots\times \mathcal{I}_{n_v}}\left(h\left(\sum\nolimits_{j=1}^{n_v}\omega_{ij}v_j + \theta_i\right)\right) = h (\underline{z}_i + \theta_i),
		\\
		\max_{\mathbf{v} \in \mathcal{I}_{1} \times\cdots\times \mathcal{I}_{n_v}}\left(h\left(\sum\nolimits_{j=1}^{n_v}\omega_{ij}v_j + \theta_i\right)\right) = h (\bar{z}_i + \theta_i).
		\end{align}
		
Thus, it leads to
	\begin{align}
	y_i \in [h(\underline{z}_i + \theta_i) ,~h (\bar{z}_i +  \theta_i) ] = \mathcal{I}_i.
	\end{align}
	and therefore, $\mathbf{y} \in \mathcal{I}_{1} \times\cdots\times \mathcal{I}_{n_y}$. The proof is complete.
\end{proof}

Theorem \ref{thm1} not only demonstrates the output set of one single layer can be approximated by a cell if the input set is a cell, it also gives out an efficient way to calculate the cell, namely by (\ref{thm1_1}) and (\ref{thm1_2}). For multi-layer neural networks, Theorem \ref{thm1}  plays the key role for the layer-by-layer approach. For an MLP which essentially has $\mathbf{v}^{[\ell]}=\mathbf{y}^{[\ell-1]}$, $\ell=1,\ldots,L$, if the input set is a set of cells, Theorem \ref{thm1} assures the input set of every layer can be over-approximated by a set of cells, which can be computed by  (\ref{thm1_1}) and (\ref{thm1_2}) layer-by-layer. The output set of layer $L$ is thus an over-approximation of reachable set of the MLP. 

Function \texttt{reachMLP} given in Algorithm \ref{alg2} illustrates the layer-by-layer method for reachable set estimation for an MLP.  

\begin{algorithm}
	\caption{Reachable set estimation for MLP} \label{alg2}
	\begin{algorithmic}[1]
		\Require Weight matrices $\mathbf{W}^{[\ell]}$, bias $\boldsymbol{\theta}^{[\ell]}$, $\ell \in \{ 1,\ldots,L\}$, set $\mathcal{V}$, partition numbers $M_i$, $i \in \{ 1,\ldots,n\}$
		\Ensure Reachable set estimation $\tilde{\mathcal{Y}}$.
		
		\Function{reachMLP}{$\mathbf{W}^{[\ell]}$, $\boldsymbol{\theta}^{[\ell]}$, $\ell \in \{ 1,\ldots,L\}$, $\mathcal{V}$, $M_i$, $i \in \{ 1,\ldots,n\}$}
		
		\State $\mathscr{P} \gets \mathrm{cell}(\mathcal{V}, M_i, i \in \{ 1,\ldots,n\})$
		
		\For{$p=1:1:\left|\mathscr{P}\right|$}
		\State $\mathcal{I}_1^{[1]}\times\cdots\times\mathcal{I}_{n^{[1]}}^{[1]} \gets \mathcal{P}_p$
		\For{$j=1:1:L$}
		\For{$i=1:1:n^{[j]}$}
		\State $\underline{g}_{ij}  \gets \left\{ {\begin{array}{*{20}c}
			{\omega _{ij} \underline{v}_j } & {\omega _{ij}  \geq 0}  \\
			{\omega _{ij} \bar v_j } & {\omega _{ij}  < 0}  \\
			\end{array} } \right.,~\bar g_{ij}  \gets \left\{ {\begin{array}{*{20}c}
			{\omega _{ij} \bar v_j } & {\omega _{ij}  \geq 0}  \\
			{\omega _{ij} \underline{v}_j } & {\omega _{ij}  < 0}  \\
			\end{array} } \right.$
	   \State $\underline{z_i} \gets \sum\nolimits_{j=1}^{n_v}\underline{g}_{ij}$, $\bar{z_i} \gets \sum\nolimits_{j=1}^{n_v}\bar{g}_{ij}$
	   \State 	$\mathcal{I}_{i}^{[j+1]} \gets 	[h_j (\underline{z}_i + \theta_i) ,~h_j (\bar{z}_i +  \theta_i) ]$		
		\EndFor		
		\EndFor
		\State $\tilde{\mathcal{Y}}_p \gets \mathcal{I}_1^{[L]}\times\cdots\times\mathcal{I}_{n^{[L]}}^{[L]}$
		\EndFor
		\State $\tilde{\mathcal{Y}} \gets \bigcup\nolimits_{p=1}^{\left|\mathscr{P}\right|}\tilde{\mathcal{Y}}_p$
		\State \Return $\tilde{\mathcal{Y}}$
		\EndFunction
	\end{algorithmic}
\end{algorithm}

\begin{example}\label{example_1}
	An MLP with 2 inputs, 2 outputs and 1 hidden layer consisting of 5 neurons is considered. 
	The activation function for the hidden layer is choosen as \texttt{tanh} function and \texttt{purelin} function is for the output layer.	
	The weight matrices and bias vectors are given as below: 
	\begin{align*}
	&\mathbf{W}^{[1]}=\left[ {\begin{array}{*{20}c}
		0.2075 &  -0.7128 \\
		0.2569 &    0.7357  \\
		-0.6136 &    -0.3624  \\
		0.0111  &   0.1393   \\
		-1.0872  &  -0.2872  \\
		\end{array} } \right],~\boldsymbol{\theta}^{[1]}=\left[ {\begin{array}{*{20}c}
		-1.1829 \\
		-0.6458 \\
		0.4619 \\
		-0.0499 \\
		0.3405\\
		\end{array} } \right],
	\\
	&\mathbf{W}^{[2]}=\left[ {\begin{array}{*{20}c}
		-0.5618 &  -0.0851 &  -0.4529 &  -0.8230   & 0.5651 \\
		0.7861 &  -0.0855 &  1.1041 &   1.6385 &  -0.3859
		\\
		\end{array} } \right],~
	\boldsymbol{\theta}^{[2]}=\left[ {\begin{array}{*{20}c}
		-0.2489 \\
		-0.1480 \\
		\end{array} } \right].
	\end{align*}  
	
	In this example, the input set is considered as below:
	\begin{align*}
	\mathcal{V}=\{\mathbf{v} \in \mathbb{R}^{2} \mid \left\|\mathbf{v}\right\|_{\infty} \le 1\}.
	\end{align*}
	
	Then, the partition numbers are chosen to be $M_1 =M_2 =20$, which means there are in total $400$ cells, $\mathcal{P}_i$, $i \in \{1,\ldots,400\}$, produced for the reachable set estimation.

	Executing function \texttt{reachMLP} for input set $\mathcal{V}$, the estimated output reachable set is given in Figure \ref{reach_1}, in which it can be seen that 400 reachtubes are obtained and the union of them is the over-approximation of reachable set. 
	
	\begin{figure}
		\begin{center}
			\includegraphics[width=12cm]{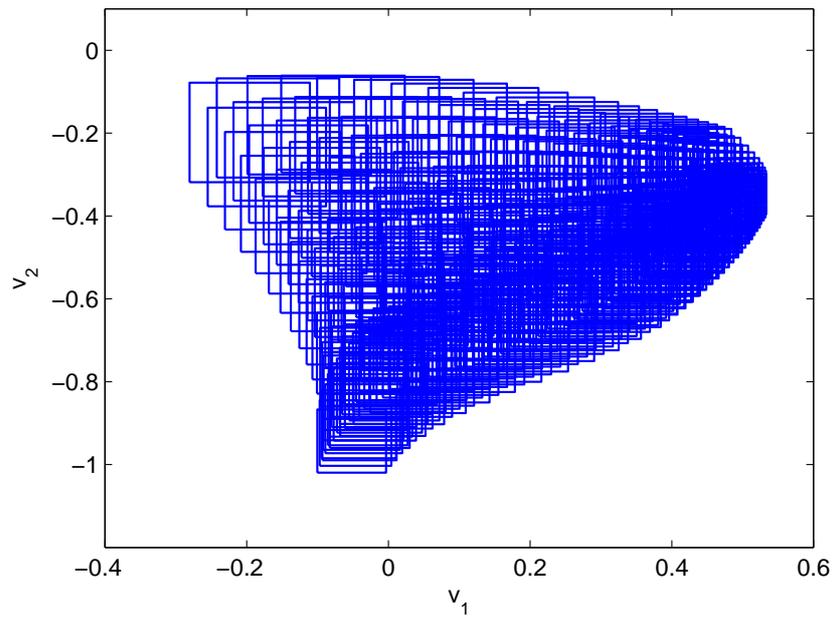}
			\caption{Output reachable set estimation with input set $\mathcal{V}=\{\mathbf{v} \in \mathbb{R}^{2} \mid \left\|v\right\|_{\infty} \le 1\}$ and partition number $M_1 =M_2 =20$. 400 reachtubes are computed for the reachable set estimation of the MLP. }
			\label{reach_1}
		\end{center}
	\end{figure}
	Moreover, we choose a different partition numbers discretizing state space to show how the choice of partitioning input set affects the estimation outcome. Explicitly,  larger partition numbers will produce more cells and generate preciser approximations of input sets and are supposed to achieve preciser estimations. Here, we adjust the partition numbers from $10$ to $50$ for the different estimation results. With this finer discretization, more computation efforts are required for running function \texttt{reachMLP}, but a tighter estimation for the  reachable set can be obtained. The  reachable set estimations are shown in Figure \ref{reach_2}. Comparing those results, it can be observed that larger partition numbers can lead to a better estimation result at the expense of more computation efforts. The computation time and  number of reachtubes  with different partition numbers are listed in Table \ref{tab_1}. 
	
		\begin{figure}
			\begin{center}
				\includegraphics[width=12cm]{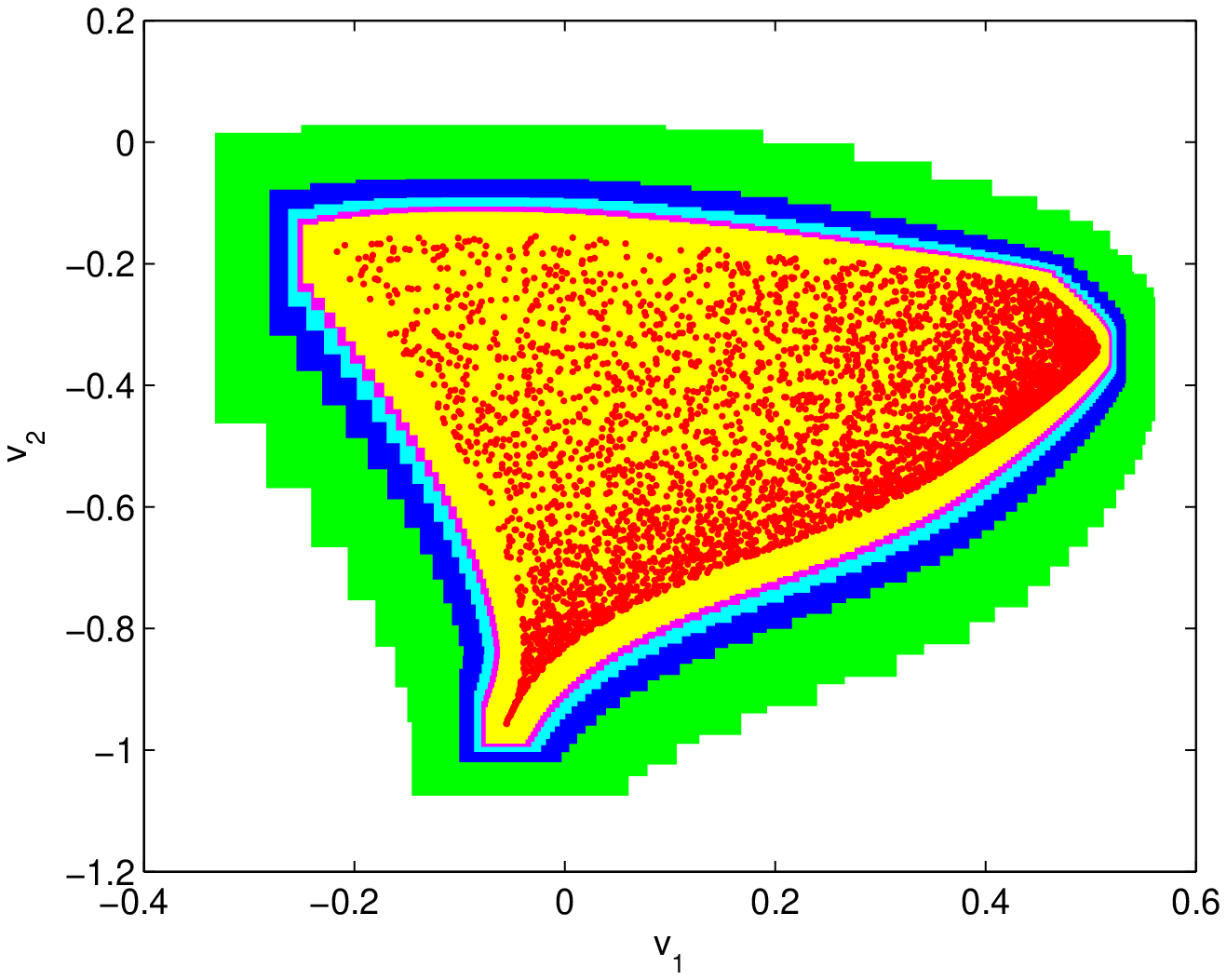}
				\caption{Output reachable set estimation with input set $\mathcal{V}=\{\mathbf{v} \in \mathbb{R}^{2} \mid \left\|v\right\|_{\infty} \le 1\}$ and partition number $M_1 =M_2 =10$ (green + blue + cyan + magenta + yellow), $M_1 =M_2 =20$ (blue + cyan + magenta + yellow), $M_1 =M_2 =30$ (cyan + magenta + yellow), $M_1 =M_2 =40$ (magenta + yellow) and $M_1 =M_2 =50$ (yellow). It can be observed that tighter estimations can be obtained with larger partition numbers. $5000$ random outputs (red spots) from input set  are all located in the estimated reachable set.}
				\label{reach_2}
			\end{center}
		\end{figure}

	\begin{table}[ht!]
		\centering
		\caption{Computation time and number of reachtubes  with different partition numbers}\label{tab_1}
		\begin{tabular}{c||c||c}
			\hline
	    	Partition Number & Computation Time  & Number of Reachtubes\\ 
			\hline
			$M_1=M_2=10$  & 0.062304 seconds & 100\\
			\hline
			 $M_1=M_2=20$ &  0.074726 seconds & 400\\
			 \hline
			  $M_1=M_2=30$&  0.142574 seconds& 900\\
			  \hline
     			 $M_1=M_2=40$ &   0.251087 seconds & 1600\\
			 \hline
			 $M_1=M_2=50$ &  0.382729 seconds & 2500\\
			\hline
			
		\end{tabular}
	\end{table} 
	
To validate the result, $5000$ random outputs are generated, it is clear to see in Figure \ref{reach_2} that all the outputs are included in the estimated reachable set, showing the effectiveness of the proposed approach. 
\end{example}

\section{Reachable Set Estimation for NARMA Models}
Based on the developed approach for reachable set estimation for MLP, this section will extend the result to NARMA models. As in previous sections, NARMA models employ MLP to approximate the nonlinear relation between $\mathbf{x}(k),\mathbf{x}(k-1),\ldots,\mathbf{x}(k-d_x),\mathbf{u}(k),\mathbf{u}(k-1),\ldots,\mathbf{u}(k-d_u)$ and state $\mathbf{x}(k+1)$. Without loss of generality, we assume $d_x \ge d_u$, thus the model is valid for any $k \ge d_x$. Thus, with the aid of reachable set estimation results for MLP, the reachable set of NARMA (\ref{NARMA}) at time instant $k$ can be estimated by recursively using functions \texttt{cell} and \texttt{reachMLP} for $k-d_x$ times.  

Since the reachable sets $\mathcal{X}_k$, $k \in \{0,1,\ldots,d_x\}$, are given as initial set, let us start with $k = d_x+1$. In the employment of function \texttt{reachMLP} with input of $\mathcal{X}_0 \times \ldots \times \mathcal{X}_{d_x}$ and $\mathcal{U}^{d_u}$, $	\tilde{\mathcal{X}}_{d_x+1} = \mathrm{reachMLP}( \mathbf{W}^{[\ell]}, \boldsymbol{\theta}^{[\ell]}, \ell \in \{ 1,\ldots,L\},\mathcal{X}_0 \times \ldots \times \mathcal{X}_{d_x}, M_i, i \in \{ 1,\ldots,n^{[0]}\})$ is an over-approximation of $\mathcal{X}_{d_x+1}$, namely $\mathcal{X}_{d_x+1} \subseteq \tilde{\mathcal{X}}_{d_x+1}$. Then, repeating using function \texttt{reachMLP} from $d_x+1$ to $k_f$, we can obtain an over-approximation of $\mathcal{X}_k$, $k = d_x+1,\ldots,k_f$, and $\mathcal{X}_{[0,k_f]}$.

\begin{proposition}
	Consider NARMA model (\ref{NARMA}) with initial set $\mathcal{X}_0 \times \ldots \times \mathcal{X}_{d_x}$ and input set $\mathcal{U}$, the reachable set $\mathcal{X}_k$, $k > d_x$ can be recursively over-approximated by 
	\begin{align}
	\tilde{\mathcal{X}}_k = &\mathrm{reachMLP}( \mathbf{W}^{[\ell]}, \boldsymbol{\theta}^{[\ell]}, \ell \in \{ 1,\ldots,L\}, \nonumber
	\\
	& \tilde{\mathcal{X}}_{k-d_x-1} \times\ldots \times \tilde{\mathcal{X}}_{k-1} \times\mathcal{U}^{d_u}, M_i, i \in \{ 1,\ldots,n^{[0]}\}),
	\end{align}
	where $\tilde{\mathcal{X}}_k = \mathcal{X}_k$, $k \in \{0,\ldots,d_x\}$. 
	the reachable set over time interval $[0, k_f]$ can be estimated by
	\begin{equation}\label{}
	\tilde{\mathcal{X}}_{[0,k_f]} = \bigcup\nolimits_{s=0}^{k_f}\tilde{\mathcal{X}}_{s}.
	\end{equation} 	
\end{proposition}

The iterative algorithm for estimating reachable set $\mathcal{X}_k$ and $\mathcal{X}_{k_f}$ is summarized as function \texttt{reachNARMA} in Algorithm \ref{alg3}.  

\begin{algorithm}
	\caption{Reachable set estimation for NARMA model} \label{alg3}
	\begin{algorithmic}[1]
		\Require  Weight matrices $\mathcal{W}^{[\ell]}$, bias $\boldsymbol{\theta}^{[\ell]}$, $\ell = 1, \ldots,L$, initial set $\mathcal{X}_0 \times \ldots \times \mathcal{X}_{d_x}$, input set $\mathcal{U}$, partition numbers $M_i, i \in \{ 1,\ldots,n^{[0]}\}$
		\Ensure Reachable set estimation $\mathcal{X}_k$, $\mathcal{X}_{[0,k_f]}$.
		
		\Function{reachNARMA}{$\mathcal{W}^{[\ell]}$, $\boldsymbol{\theta}^{[\ell]}$, $\ell = 1, \ldots,L$, $\mathcal{X}_0 \times \ldots \times \mathcal{X}_{d_x}$, $\mathcal{U}$, $M_i, i \in \{ 1,\ldots,n^{[0]}\}$}
		\For{$k =d_u+1:1:k_f$}
		\State $\mathcal{V} \gets  \mathcal{X}_{k-d_u-1} \times\ldots \times \mathcal{X}_{k-1} \times \mathcal{U}^{d_u}$
		\State $\mathcal{X}_k \gets \mathrm{reachMLF}(\mathcal{W}^{[\ell]}, \boldsymbol{\theta}^{[\ell]}, \ell = 1, \ldots,L, \mathcal{V},M_i, i \in \{ 1,\ldots,n^{[0]}\}$.
		
		\EndFor
		
		\State 	$\mathcal{X}_{[0,k_f]} \gets \bigcup\nolimits_{s=0}^{k_f}\mathcal{X}_{s}$
		\State \Return $\mathcal{X}_k$, $k=0,1\ldots,k_f$, $\mathcal{X}_{[0,k_f]}$
		\EndFunction
	\end{algorithmic}
\end{algorithm}

Function \texttt{reachNARMA} is sufficient to solve the reachable set estimation problem for an NARMA model, that is Problem \ref{problem1}. Then,
we can move forward to Problem \ref{problem2}, the safety verification
problem for an NARMA model with a given safety specification $\mathcal{S}$ over a finite interval $[0,k_f]$, with the aid of estimated reachable set $\mathcal{X}_{[0,k_f]}$. Given a safety specification $\mathcal{S}$, the empty intersection between over-approximation $\tilde{\mathcal{X}}_{[0,k_f]}$ and $\neg \mathcal{S}$, namely $\tilde{\mathcal{X}}_{[0,k_f]} \cap \neg \mathcal{S} =\emptyset$, naturally leads to $\mathcal{X}_{[0,k_f]} \cap \neg \mathcal{S} =\emptyset$ due to $\mathcal{X}_{[0,k_f]}\subseteq\tilde{\mathcal{X}}_{[0,k_f]}$.  The safety verification result is summarized by the following proposition.

\begin{proposition}\label{proposition_2}
	Consider NARMA model (\ref{NARMA}) with initial set $\mathcal{X}_0 \times \ldots \times \mathcal{X}_{d_x}$,  input set $\mathcal{U}$, and a safety specification $\mathcal{S}$, the NARMA model
	 (\ref{NARMA}) is safe in interval $[0,k_f]$, if $ \tilde{\mathcal{X}}_{[0,k_f]} \cap \neg \mathcal{S} = \emptyset$, where $\tilde{\mathcal{X}}_{[0,k_f]} = \mathrm{reachNARMA}(\mathcal{W}^{[\ell]}, \boldsymbol{\theta}^{[\ell]}, \ell = 1, \ldots,L, \mathcal{X}_0 \times \ldots \times \mathcal{X}_{d_x}, \mathcal{U}, M_i, i \in \{ 1,\ldots,n^{[0]}\})$  obtained by Algorithm \ref{alg3}.
\end{proposition}

Function \texttt{verifyNARMA} is developed based on Proposition \ref{proposition_2} for Problem \ref{problem2}, the safety verification problem for NARMA model.  If function \texttt{verifyNARMA} returns SAFE then the NARMA model is safe. If it returns UNCERTAIN, caused by the fact $\tilde{\mathcal{X}}_{[0,k_f]}$, that means the safety property is unclear for this case. 

\begin{algorithm}
	\caption{Safety verification for NARMA model} \label{alg4}
	\begin{algorithmic}[1]
		\Require  Weight matrices $\mathcal{W}^{[\ell]}$, bias $\boldsymbol{\theta}^{[\ell]}$, $\ell = 1, \ldots,L$, initial set $\mathcal{X}_0 \times \ldots \times \mathcal{X}_{d_x}$, input set $\mathcal{U}$, partition numbers $M_i, i \in \{ 1,\ldots,n^{[0]}\}$, safety specification $\mathcal{S}$
		\Ensure SAFE or UNCERTAIN.
		
		\Function{verifyNARMA}{$\mathcal{W}^{[\ell]}$, $\boldsymbol{\theta}^{[\ell]}$, $\ell = 1, \ldots,L$, $\mathcal{X}_0 \times \ldots \times \mathcal{X}_{d_x}$, $\mathcal{U}$, $M_i, i \in \{ 1,\ldots,n^{[0]}\}$, $\mathcal{S}$}
		
		\State  $\mathcal{X}_[0,k_f] \gets \mathrm{reachNARMA}(\mathcal{W}^{[\ell]},\boldsymbol{\theta}^{[\ell]},\ell = 1, \ldots,L,\mathcal{X}_0 \times \ldots \times \mathcal{X}_{d_x},\mathcal{U},M_i, i \in \{ 1,\ldots,n^{[0]}\})$
		\If{$\mathcal{X}_[0,k_f] \cap \mathcal{S} = \emptyset$}
		\State \Return SAFE
		\Else 
		\State \Return UNCERTAIN
		\EndIf
		\EndFunction
	\end{algorithmic}
\end{algorithm}

An numerical example is provided to show the effectiveness of our developed approach.

\begin{example}
	In this example, we consider an NARMA model as below:
	\begin{equation}
	\mathbf{x}(k+1) = f(\mathbf{x}(k),\mathbf{u}(k)),
	\end{equation}
	where $\mathbf{x}(k),\mathbf{u}(k) \in \mathbb{R}$. We use an MLP with 2 inputs, 1 outputs and 1 hidden layer consisting of 5 neurons to approximate $f$ with weight matrices and bias vectors below: 
	\begin{align*}
	&\mathbf{W}^{[1]}=\left[ {\begin{array}{*{20}c}
	0.1129 &   0.4944 \\
	2.2371  &  0.4389 \\
	-1.1863  & -0.7365 \\
	0.2965  &  0.3055 \\
	-0.6697  &  0.5136 \\
		\end{array} } \right],~\boldsymbol{\theta}^{[1]}=\left[ {\begin{array}{*{20}c}
	 -13.8871 \\
	 -8.2629 \\
	 5.8137 \\
	 -3.2035 \\
	 -0.6697 \\
		\end{array} } \right],
	\\
	&\mathbf{W}^{[2]}=\left[ {\begin{array}{*{20}c}
		 -3.3067  &  1.3905  & -0.6422 &   2.5221    & 1.8242
		\\
		\end{array} } \right],~
	\boldsymbol{\theta}^{[2]}=\left[ {\begin{array}{*{20}c}
		5.8230 \\
		\end{array} } \right]
	\end{align*} 
    
The activation function for the hidden layer is choose \texttt{tanh} function and \texttt{purelin} function is for the output layer.	
The initial set and input set are given by the following set
\begin{align}
\mathcal{X}_0 &= \{\mathbf{x}(0) \in \mathbb{R} \mid - 0.2 \le \mathbf{x}(0) \le 0.2\}, \\
\mathcal{U}& = \{\mathbf{u}(k) \in \mathbb{R} \mid 0.8 \le \mathbf{u}(k) \le 1.2,~\forall k \in \mathbb{N} \}.
\end{align}

We set the partition numbers to be $M_1 = M_2 =10$, where $M_1$ is for input $\mathbf{u}$ and $M_2$ is for state $\mathbf{x}$. The time horizon for the reachable set estimation is set to be $ [0, 50]$. Using function \texttt{reachNARMA}, the reachable set can be estimated, which is shown in Figure \ref{reach_4}. 	To show the effectiveness of our proposed approach, we randomly generate 100 state trajectories that are all within the estimated reachable set. 

Furthermore, with the estimated reachable set, the safety verification can be easily performed. For example, if the safety region is assumed to be $\mathcal{S} = \{\mathbf{x}\in \mathbb{R}\mid \mathbf{x} \le 16\}$, it is easy to verify that $\tilde{\mathcal{X}}_{[0,50]} \cap \neg \mathcal{S} = \emptyset$ which means the NARMA model is safe. 
	
\begin{figure}[h!]
	\begin{center}
	\includegraphics[width=12cm]{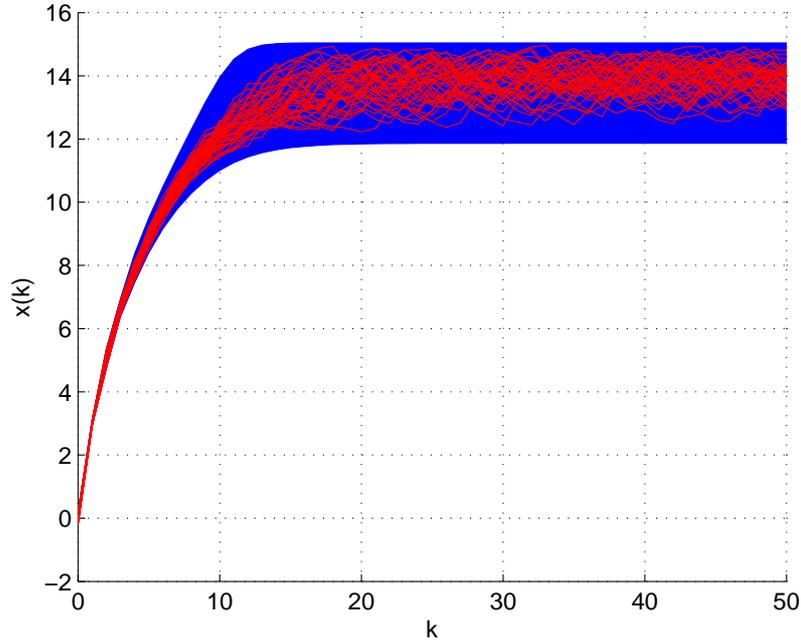}
	\caption{Reachable set estimation for NARMA model. Blue area is the estimated reachable set and red solid lines are 100 randomly generated state trajectories. All the randomly generated state trajectories are in the reachable set estimation area. }
		\label{reach_4}
		\end{center}
	\end{figure}
\end{example}

\section{Magnetic Levitation Systems (Maglev)}
\subsection{Brief Introduction}
\par Magnetic Levitation Systems, which are called Maglev Systems in short,  are systems in which an object is suspended exclusively by the presence of magnetic fields. In such schemes, the force exerted by the presence of magnetic fields is able to counteract gravity and any other forces acting on the object \cite{939074}. In order to achieve levitation, there are two principle concerns. The first concern is to exert a sufficient lifting force with which to counteract gravity and the second concern is stability. Once levitation has been achieved, it is critical to ensure that the system does not move into a configuration in which the lifting forces are neutralized \cite{996030}. However, attaining stable levitation is a considerably complex task, and in his famous theorem, Samuel Earnshaw demonstrated that there is no static configuration of stability for magnetic systems \cite{Duffin1963}. Intuitively, the instability of magnetic systems lies in the fact that magnetic attraction or repulsion increases or decreases in relation to the square of distance. Thus, most control strategies  for  Maglev Systems make use of servo-mechanisms \cite{4696056} and a feedback linearization \cite{6819156} around a particular operating point of the complex nonlinear differential equations \cite{6852423} describing the sophisticated mechanical and electrical dynamics. Despite their intrinsic complexity, these systems have exhibited utility in numerous contexts and in particular  Maglev System have generated considerable scientific interest in transportation due to their ability to minimize mechanical loss, allow faster travel \cite{6936124}, minimize mechanical vibration, and emit low levels of noise\cite{1322129}. Other application domains of such systems include wind tunnel levitation \cite{4696056}, contact-less melting, magnetic bearings, vibrator isolation systems, and rocket-guiding designs \cite{937416}. Consequently,  Maglev Systems have been  extensively studied in control literature \cite{939074}.
\par Due to their unstable, complex, and nonlinear nature, it is difficult to build a precise feedback control model for the dynamic behavior of complex  Maglev System. In most cases, a linearization of the nonlinear dynamics is susceptible to a great deal of inaccuracy and uncertainty. As the system deviates from an assumed operating point, the accuracy of the model deteriorates \cite{7840137}. Additionally, models based on simplifications are often unable to handle the presence of disturbance forces. Thus, to improve control schemes, a stricter adherence to the complex nonlinear nature of the  Maglev Systems is needed. In the last several years, neural network control systems have received significant attention due to their ability to capture complex nonlinear dynamics and model nonlinear unknown parameters \cite{6852423}. 

In the control of magnetic levitation systems the nonlinear nature can be modeled by a neural network that is able to describe the input-output nature of  the nonlinear dynamics \cite{4696056}. Neural networks have shown the ability to approximate any nonlinear function to any desired accuracy \cite{1519654}. Using the a neural network model of the  plant we wish to control, a controller can be designed to meet system specifications. While neural control schemes have been successful in creating stable controllers for nonlinear systems, it is essential to demonstrate that these systems do not enter undesirable states. As an example, in the requirements for a  Maglev train system developed in 1997 by the Japanese Ministry of transportation, the measurements of the 500 km/h train's position and speed could deviate by a maximum of 3 cm and 1 km/h, respectively, in order to prevent derailment and contact with the railway \cite{988732}. As magnetic systems become more prevalent in transportation and in other domains,  the verification of these systems is essential. Thus, in this example, we perform a reachable set estimation of a NARMA neural network model (\ref{NARMA}) of a  Maglev System.

\subsection{Neural Network Model}
\par The  Maglev System we consider consists of a magnet suspended above an electromagnet where the magnet is confined to only moving in the vertical direction \cite{939074}. Using the results of De J\'esus et. al \cite{939074}, the nonlinear equation of motion for the system is 
\begin{equation}
\frac{d^2y(t)}{dt^2}=-g+\frac{\alpha}{M}\frac{i^2(t)}{y(t)}-\frac{\beta}{M}\frac{dy(t)}{dt},
\end{equation}
where $y(t)$ is the vertical position of the magnet above the electromagnet in $mm$, $i(t)$, in Amperes, is the current flowing in the electromagnet, $M$ is the mass of the magnet, $g$ is the gravitational constant, $\beta$ is the frictional coefficient, and $\alpha$ is the field strength constant. The frictional coefficient $\beta$ is dictated by the material in which the magnet moves. In our case, the magnet moves through air. The field strength constant $\alpha$ is determined by the number of turns of wire in our electromagnet and by the strength of the magnet being levitated \cite{939074}. 
\par To capture the nonlinear input-output dynamics of the system, we trained a NARMA neural network (\ref{NARMA}) to predict the magnet's future position values. In order to predict the magnet's future position values, two inputs are supplied to the network: the first is the past value of  the current flowing in the electromagnet $i(k-1)$ and the second input is the magnet's previous position value $y(k-1)$. The output of the neural network is the current position $y(k)$. The network consists of one hidden layer with eight neurons and an output layer with one neuron. The transfer function of the first layer is \texttt{tanh} and \texttt{purelin} for the output layer.

The network is trained using a data set consisting of 4001 target position values for the output and 4001 input current values. The Levenberg-Marquard algorithm \cite{847778} is used to train the network using batch training. Using batch training, the weights and biases of the NARMA model (\ref{NARMA}) are updated after all the inputs and targets are supplied to the network and a gradient descent algorithm is used to minimize error \cite{Beale12neuralnetwork}. To avoid over-fitting the network, the training data is divided randomly into three sets: the training set, which consists of 2801 values, the validation set, which consists of 600 values,  and a test set which is the same size as the validation set. The training set is used to adjust the weight and bias values of the network as well as to compute the gradient, while the validation set is used to measure the network's generalization. Training of the networks ceases when the network's generalization to input data stops improving. The testing data does not take part into the training, but it is used to check the performance of the net during and after training.

In this example, we set the minimum gradient to $10^{-7}$, and set the number of validation checks to 6. Thus, training ceases if the error on the validation set increases for 6 consecutive iterations or the minimum gradient achieves a value of $10^{-7}$. In our case, the training stopped when the validation checks reached its limit of 6, obtaining a performance of 0.000218. 
Initially, before the training begins, the values of the weights, biases, and training set are initialized randomly. Thus, the value of the weights and the biases may be different every time that the network is trained. The weights and biases of the hidden layer are 
\[W^{[1]} =\begin{bmatrix}
      -68.9367 & -3.3477 \\[0.3em]
      -0.0802 & -2.1460 \\[0.3em]
      0.1067 & -3.7875 \\[0.3em]
      0.1377 & -1.5763 \\[0.3em]
      -0.3954 & -1.4477 \\[0.3em]
      -0.4481 & -6.9485 \\[0.3em]
      0.0030 & 1.5819 \\[0.3em]
      5.9623 & -5.5775         
\end{bmatrix},
\boldsymbol{\theta}^{[1]} = \begin{bmatrix}
      47.8492 \\[0.3em]
      2.2129 \\[0.3em]
      1.9962 \\[0.3em]
      -0.0091 \\[0.3em]
      -0.0727 \\[0.3em]
      -3.8435 \\[0.3em]
      1.7081 \\[0.3em]
      7.5619 
\end{bmatrix}\]
and in the output layer, the weights and the biases are
\[W^{[2]} =\begin{bmatrix}
      -0.0054  & -0.3285 &  -0.0732 &  -0.4019 &  -0.1588 &  -0.0128  &  0.5397 &  -0.0279   
\end{bmatrix},\]
\[\boldsymbol{\theta}^{[2]} = \begin{bmatrix}
      0.1095
\end{bmatrix}.\]

Once the NARMA network model (\ref{NARMA}) is trained and the weight and bias values are adjusted to the values shown above, the reachable set estimation of the system can be computed and a safety requirement $\mathcal{S}$ could be verified. This computation is executed following the process described in the previous section. 

\subsection{Reachable Set Estimation}
In order to compute the reachable set and verify if the given specification is satisfied, Algorithm \ref{alg3} is employed. First, the reachable set estimation using 5 partitions is computed, followed by the reachable set estimation using 20 partitions. After both reachable set estimations are calculated, 200 random trajectories are generated and plotted into Figure \ref{5partitions}. 

The reachable set estimations and the random trajectories are computed with an initial set and input set that are assumed to be given by 
\begin{align}
\mathcal{X}_0 &= \{\mathbf{x}(0) \in \mathbb{R} \mid 4.00 \le \mathbf{x}(0) \le 5.00\}, \\
\mathcal{U}& = \{\mathbf{u}(k) \in \mathbb{R} \mid 0.10 \le \mathbf{u}(k) \le 1.10,~\forall k \in \mathbb{N} \}.
\end{align}

\begin{figure}[h!]
      \begin{center}
 \includegraphics[width=1\textwidth,height=0.5\textheight]{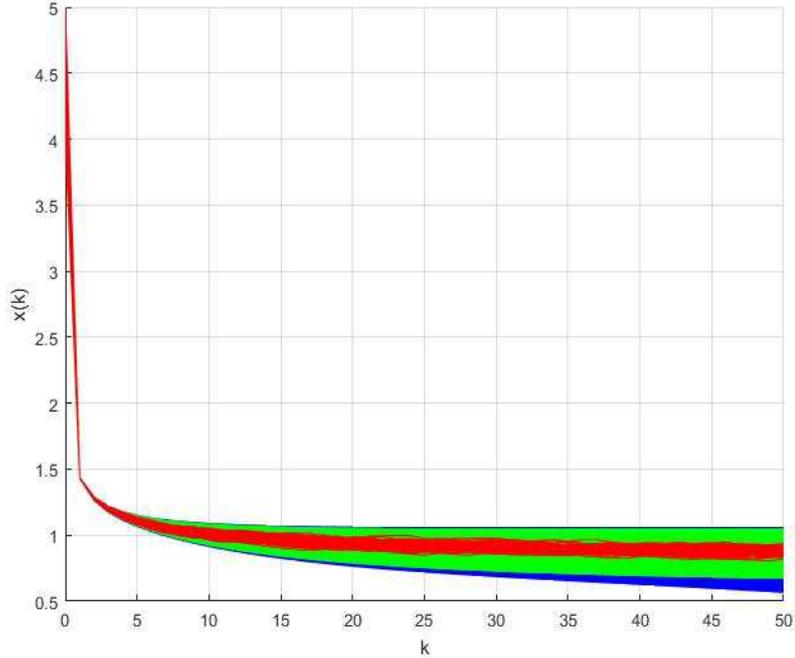} 
        \caption{Reachable set estimation using 5 and 20 partitions. The blue area corresponds to the estimated reachable set using 5 partitions, the tighter green area corresponds to the reachable set estimation using 20 partitions, and the red lines correspond to 200 randomly generated state trajectories, which all of them lie within the estimated reachable set area.}        \label{5partitions}
      \end{center}
\end{figure}
As is observed from Figure \ref{5partitions}, all the randomly generated trajectories lie within the estimated reachable set. Also, it can be noted that the area of the reachable set estimation using a larger partition number, that is 20, represented in green,  it is smaller than the blue area, which corresponds to the reachable set estimation using a lower partition number ($M_1=M_2=5$) . This is especially noticeable as the time $k$ increases to 40--50, where the difference between the blue region and green region increases as the lower limit of the state $\mathbf{x}(k)$ using 5 partitions keeps decreasing towards $0.6$, while the lower limit of the green area maintains a more steady line at $0.7$ approximately.

	\begin{table}[ht!]
		\centering
		\caption{Computational time for different partition numbers}\label{tab_2}
		\begin{tabular}{c||c}
			\hline
	    	Partition Number & Computation Time \\
			\hline
			$M_1=M_2=5$  & 0.048700 seconds \\
			\hline
			 $M_1=M_2=20$ &  0.474227 seconds \\
			\hline
		\end{tabular}
	\end{table} 
In Table \ref{tab_2}, the computational time has been recorded for each reachable set estimation. It can be observed that the computational time increases as the partition number increases. For this system, the computational time is approximately 10 times greater when 20 partitions are used. This means that every approach has its different advantages. For the cases when a more precise estimation is needed, we can increase the number of partitions, while for the cases when an larger over-approximation is enough, the number of partitions may be decreased to reduce its computational cost. 

The reachable set estimation for the NARMA neural network model (\ref{NARMA}) of the Maglev Systems shows that all system responses to inputs are contained within the reachable set. Thus, our over-approximation of the reachable states is valid. Given a safety specification $\mathcal{S}$ and the reachable set calculated using Algorithm \ref{alg3}, we are able to determine whether our system model satisfies $\mathcal{S}$. In our example, we did not perform a safety analysis but rather demonstrated the robustness of Algorithm \ref{alg3} in capturing a large number of possible predictions of the NARMA network model (\ref{NARMA}). 
The magnet in our example was confined to moving in one dimension. In magnetic levitation systems that are not physically constrained to a set of axes, there are six degrees of freedom  (three rotational and three transnational) \cite{BerkelmanLorentz}. Thus, while we have demonstrated that our algorithm is robust for two-dimensional systems, it will be good to demonstrate its efficacy on higher dimensional systems. However, as the dimensionality and size of the neural networks increases, the computation time needed to compute the reachable set increases significantly as well. 
 
\section{Conclusions}
This paper studies the reachable set estimation problem for neural network NARMA model of nonlinear dynamic systems. By partitioning the input set into a finite number of cells, the reachable set estimation for MLPs can be done for each individual cells and get the union of output set of cells to form an over-approximation of output set. Then, the reachable set estimation for NARMA models can be performed by iterating the reachable set estimation process for MLP step by step to establish an estimation for state trajectories of an NARMA model. Safety property of an NARMA model can be verified by checking the intersection between the estimated reachable set and unsafe regions. The approach is demonstrated by a Maglev System, for which the reachable set of its NARMA neural network model is estimated. The approach is applicable for a variety of neural network models with different activation functions. However, since the estimation is an over-approximation and the error will accumulate at each layer, much finer discretization for input space is required for deep neural networks that essentially have large number of layers, which will introduce large computation effort, otherwise the estimation results will be too conservative. Reducing the conservativeness caused by the increase of layers and generalizing it to deep neural network will be our future focus for our approach. 

\subsubsection*{Acknowledgments.} The material presented in this paper is based upon work supported by
the National Science Foundation (NSF) under grant numbers CNS 1464311,
EPCN 1509804, and SHF 1527398, the Air Force Research Laboratory
(AFRL) through contract numbers FA8750-15-1-0105 and FA8650-12-3-
7255 via subcontract number WBSC 7255 SOI VU 0001, and the Air Force
Office of Scientific Research (AFOSR) under contract numbers FA9550-15-
1-0258 and FA9550-16-1-0246. The U.S. government is authorized to reproduce
and distribute reprints for Governmental purposes notwithstanding
any copyright notation thereon. Any opinions, findings, and conclusions or
recommendations expressed in this publication are those of the authors and
do not necessarily reflect the views of AFRL, AFOSR, or NSF.

\bibliographystyle{plain} 
\bibliography{ref,maglev}

\begin{thebibliography}{10}

\bibitem{7840137}
J.~I. Baig and A.~Mahmood.
\newblock Robust control design of a magnetic levitation system.
\newblock In {\em 2016 19th International Multi-Topic Conference (INMIC)},
  pages 1--5, Dec 2016.

\bibitem{bak2017hylaa}
Stanley Bak and Parasara~Sridhar Duggirala.
\newblock Hy\textsc{LAA}: A tool for computing simulation-equivalent
  reachability for linear systems.
\newblock In {\em Proceedings of the 20th International Conference on Hybrid
  Systems: Computation and Control}, pages 173--178. ACM, 2017.

\bibitem{bak2017rigorous}
Stanley Bak and Parasara~Sridhar Duggirala.
\newblock Rigorous simulation-based analysis of linear hybrid systems.
\newblock In {\em International Conference on Tools and Algorithms for the
  Construction and Analysis of Systems}, pages 555--572. Springer, 2017.

\bibitem{Beale12neuralnetwork}
Mark~Hudson Beale, Martin~T. Hagan, and Howard~B. Demuth.
\newblock Neural network toolbox™ user’s guide.
\newblock In {\em R2016a, The MathWorks, Inc., 3 Apple Hill Drive Natick, MA
  01760-2098, , www.mathworks.com}, 2012.

\bibitem{BerkelmanLorentz}
Peter~J. Berkelman and Ralph~L. Hollis.
\newblock Lorentz magnetic levitation for haptic interaction: Device design,
  performance, and integration with physical simulations.
\newblock {\em The International Journal of Robotics Research}, 19(7):644--667,
  2000.

\bibitem{bojarski2016end}
Mariusz Bojarski, Davide Del~Testa, Daniel Dworakowski, Bernhard Firner, Beat
  Flepp, Prasoon Goyal, Lawrence~D Jackel, Mathew Monfort, Urs Muller, Jiakai
  Zhang, et~al.
\newblock End to end learning for self-driving cars.
\newblock {\em arXiv preprint arXiv:1604.07316}, 2016.

\bibitem{Duffin1963}
R.~J. Duffin.
\newblock Free suspension and earnshaw's theorem.
\newblock {\em Archive for Rational Mechanics and Analysis}, 14(1):261--263,
  Jan 1963.

\bibitem{duggirala2015c2e2}
Parasara~Sridhar Duggirala, Sayan Mitra, Mahesh Viswanathan, and Matthew Potok.
\newblock \textsc{C2E2}: a verification tool for stateflow models.
\newblock In {\em International Conference on Tools and Algorithms for the
  Construction and Analysis of Systems}, pages 68--82. Springer, 2015.

\bibitem{fan2016automatic}
Chuchu Fan, Bolun Qi, Sayan Mitra, Mahesh Viswanathan, and Parasara~Sridhar
  Duggirala.
\newblock Automatic reachability analysis for nonlinear hybrid models with
  \textsc{C2E2}.
\newblock In {\em International Conference on Computer Aided Verification},
  pages 531--538. Springer, 2016.

\bibitem{ge1999adaptive}
Shuzhi~Sam Ge, Chang~Chieh Hang, and Tao Zhang.
\newblock Adaptive neural network control of nonlinear systems by state and
  output feedback.
\newblock {\em IEEE Transactions on Systems, Man, and Cybernetics, Part B
  (Cybernetics)}, 29(6):818--828, 1999.

\bibitem{937416}
A.~El Hajjaji and M.~Ouladsine.
\newblock Modeling and nonlinear control of magnetic levitation systems.
\newblock {\em IEEE Transactions on Industrial Electronics}, 48(4):831--838,
  Aug 2001.

\bibitem{hornik1989multilayer}
Kurt Hornik, Maxwell Stinchcombe, and Halbert White.
\newblock Multilayer feedforward networks are universal approximators.
\newblock {\em Neural Networks}, 2(5):359--366, 1989.

\bibitem{huang2016safety}
Xiaowei Huang, Marta Kwiatkowska, Sen Wang, and Min Wu.
\newblock Safety verification of deep neural networks.
\newblock {\em arXiv preprint arXiv:1610.06940}, 2016.

\bibitem{hunt1992neural}
K~Jetal Hunt, D~Sbarbaro, R~{\.Z}bikowski, and Peter~J Gawthrop.
\newblock Neural networks for control systems: a survey.
\newblock {\em Automatica}, 28(6):1083--1112, 1992.

\bibitem{939074}
O.~De Jesus, A.~Pukrittayakamee, and M.~T. Hagan.
\newblock A comparison of neural network control algorithms.
\newblock In {\em Neural Networks, 2001. Proceedings. IJCNN '01. International
  Joint Conference on}, volume~1, pages 521--526 vol.1, 2001.

\bibitem{1322129}
J.~Kaloust, C.~Ham, J.~Siehling, E.~Jongekryg, and Q.~Han.
\newblock Nonlinear robust control design for levitation and propulsion of a
  maglev system.
\newblock {\em IEE Proceedings - Control Theory and Applications},
  151(4):460--464, July 2004.

\bibitem{katz2017reluplex}
Guy Katz, Clark Barrett, David Dill, Kyle Julian, and Mykel Kochenderfer.
\newblock Reluplex: An efficient \textsc{SMT} solver for verifying deep neural
  networks.
\newblock {\em arXiv preprint arXiv:1702.01135}, 2017.

\bibitem{6936124}
C.~H. Kim, J.~Lim, J.~M. Lee, H.~S. Han, and D.~Y. Park.
\newblock Levitation control design of super-speed maglev trains.
\newblock In {\em 2014 World Automation Congress (WAC)}, pages 729--734, Aug
  2014.

\bibitem{lawrence1997face}
Steve Lawrence, C~Lee Giles, Ah~Chung Tsoi, and Andrew~D Back.
\newblock Face recognition: A convolutional neural-network approach.
\newblock {\em IEEE Transactions on Neural Networks}, 8(1):98--113, 1997.

\bibitem{1519654}
Xiao-Dong Li, J.~K.~L. Ho, and T.~W.~S. Chow.
\newblock Approximation of dynamical time-variant systems by continuous-time
  recurrent neural networks.
\newblock {\em IEEE Transactions on Circuits and Systems II: Express Briefs},
  52(10):656--660, Oct 2005.

\bibitem{847778}
L.~S.~H. Ngia and J.~Sjoberg.
\newblock Efficient training of neural nets for nonlinear adaptive filtering
  using a recursive levenberg-marquardt algorithm.
\newblock {\em IEEE Transactions on Signal Processing}, 48(7):1915--1927, Jul
  2000.

\bibitem{988732}
M.~Ono, S.~Koga, and H.~Ohtsuki.
\newblock Japan's superconducting maglev train.
\newblock {\em IEEE Instrumentation Measurement Magazine}, 5(1):9--15, Mar
  2002.

\bibitem{pulina2010abstraction}
Luca Pulina and Armando Tacchella.
\newblock An abstraction-refinement approach to verification of artificial
  neural networks.
\newblock In {\em International Conference on Computer Aided Verification},
  pages 243--257. Springer, 2010.

\bibitem{pulina2012challenging}
Luca Pulina and Armando Tacchella.
\newblock Challenging \textsc{SMT} solvers to verify neural networks.
\newblock {\em AI Communications}, 25(2):117--135, 2012.

\bibitem{996030}
D.~M. Rote and Yigang Cai.
\newblock Review of dynamic stability of repulsive-force maglev suspension
  systems.
\newblock {\em IEEE Transactions on Magnetics}, 38(2):1383--1390, Mar 2002.

\bibitem{schmidhuber2015deep}
J{\"u}rgen Schmidhuber.
\newblock Deep learning in neural networks: An overview.
\newblock {\em Neural Networks}, 61:85--117, 2015.

\bibitem{silver2016mastering}
David Silver, Aja Huang, Chris~J Maddison, Arthur Guez, Laurent Sifre, George
  Van Den~Driessche, Julian Schrittwieser, Ioannis Antonoglou, Veda
  Panneershelvam, Marc Lanctot, et~al.
\newblock Mastering the game of go with deep neural networks and tree search.
\newblock {\em Nature}, 529(7587):484--489, 2016.

\bibitem{szegedy2013intriguing}
Christian Szegedy, Wojciech Zaremba, Ilya Sutskever, Joan Bruna, Dumitru Erhan,
  Ian Goodfellow, and Rob Fergus.
\newblock Intriguing properties of neural networks.
\newblock {\em arXiv preprint arXiv:1312.6199}, 2013.

\bibitem{thuan2016reachable}
Mai~Viet Thuan, Hieu~Manh Tran, and Hieu Trinh.
\newblock Reachable sets bounding for generalized neural networks with interval
  time-varying delay and bounded disturbances.
\newblock {\em Neural Computing and Applications}, pages 1--12, 2016.

\bibitem{6819156}
R.~Uswarman, A.~I. Cahyadi, and O.~Wahyunggoro.
\newblock Control of a magnetic levitation system using feedback linearization.
\newblock In {\em 2013 International Conference on Computer, Control,
  Informatics and Its Applications (IC3INA)}, pages 95--98, Nov 2013.

\bibitem{4696056}
R.~J. Wai and J.~D. Lee.
\newblock Robust levitation control for linear maglev rail system using fuzzy
  neural network.
\newblock {\em IEEE Transactions on Control Systems Technology}, 17(1):4--14,
  Jan 2009.

\bibitem{xiang2015equivalence}
Weiming Xiang.
\newblock On equivalence of two stability criteria for continuous-time switched
  systems with dwell time constraint.
\newblock {\em Automatica}, 54:36--40, 2015.

\bibitem{xiang2016necessary}
Weiming Xiang.
\newblock Necessary and sufficient condition for stability of switched
  uncertain linear systems under dwell-time constraint.
\newblock {\em IEEE Transactions on Automatic Control}, 61(11):3619--3624,
  2016.

\bibitem{xiang2018Parameter}
Weiming Xiang.
\newblock Parameter-memorized \textsc{L}yapunov functions for discrete-time
  systems with time-varying parametric uncertainties.
\newblock {\em Automatica}, 87:450--454, 2018.

\bibitem{xiang2017stability}
Weiming Xiang, James Lam, and Jun Shen.
\newblock Stability analysis and $\mathcal{L}_1$-gain characterization for
  switched positive systems under dwell-time constraint.
\newblock {\em Automatica}, 85:1--8, 2017.

\bibitem{xiang2017robust}
Weiming Xiang, Hoang-Dung Tran, and T.~T. Johnson.
\newblock Robust exponential stability and disturbance attenuation for
  discrete-time switched systems under arbitrary switching.
\newblock {\em IEEE Transactions on Automatic Control}, 2017, doi:
  10.1109/TAC.2017.2748918.

\bibitem{xiang2017reachable}
Weiming Xiang, Hoang-Dung Tran, and Taylor~T Johnson.
\newblock On reachable set estimation for discrete-time switched linear systems
  under arbitrary switching.
\newblock In {\em American Control Conference (ACC), 2017}, pages 4534--4539.
  IEEE, 2017.

\bibitem{xiang2017output_arxiv}
Weiming Xiang, Hoang-Dung Tran, and Taylor~T Johnson.
\newblock Output reachable set estimation and verification for multi-layer
  neural networks.
\newblock {\em arXiv preprint arXiv:1708.03322}, 2017.

\bibitem{xiang2017output}
Weiming Xiang, Hoang-Dung Tran, and Taylor~T Johnson.
\newblock Output reachable set estimation for switched linear systems and its
  application in safety verification.
\newblock {\em IEEE Transactions on Automatic Control}, 62(10):5380--5387,
  2017.

\bibitem{xiang2017reachable_arxiv}
Weiming Xiang, Hoang-Dung Tran, and Taylor~T Johnson.
\newblock Reachable set computation and safety verification for neural networks
  with \textsc{R}e\textsc{LU} activations.
\newblock {\em arXiv preprint arXiv: 1712.08163}, 2017.

\bibitem{xiang2014stabilization}
Weiming Xiang and Jian Xiao.
\newblock Stabilization of switched continuous-time systems with all modes
  unstable via dwell time switching.
\newblock {\em Automatica}, 50(3):940--945, 2014.

\bibitem{xu2017reachable}
Zhaowen Xu, Hongye Su, Peng Shi, Renquan Lu, and Zheng-Guang Wu.
\newblock Reachable set estimation for \textsc{M}arkovian jump neural networks
  with time-varying delays.
\newblock {\em IEEE Transactions on Cybernetics}, 47(10):3208--3217, 2017.

\bibitem{zhang2016mode}
Lixian Zhang and Weiming. Xiang.
\newblock Mode-identifying time estimation and switching-delay tolerant control
  for switched systems: An elementary time unit approach.
\newblock {\em Automatica}, 64:174--181, 2016.

\bibitem{zhang2016synchronization}
Lixian Zhang, Yanzheng Zhu, and Wei~Xing Zheng.
\newblock Synchronization and state estimation of a class of hierarchical
  hybrid neural networks with time-varying delays.
\newblock {\em IEEE Transactions on Neural Networks and Learning Systems},
  27(2):459--470, 2016.

\bibitem{zhang2017state}
Lixian Zhang, Yanzheng Zhu, and Wei~Xing Zheng.
\newblock State estimation of discrete-time switched neural networks with
  multiple communication channels.
\newblock {\em IEEE Transactions on Cybernetics}, 47(4):1028--1040, 2017.

\bibitem{6852423}
S.~T. Zhao and X.~W. Gao.
\newblock Neural network adaptive state feedback control of a magnetic
  levitation system.
\newblock In {\em The 26th Chinese Control and Decision Conference (2014
  CCDC)}, pages 1602--1605, May 2014.

\bibitem{zuo2014non}
Zhiqiang Zuo, Zhenqian Wang, Yinping Chen, and Yijing Wang.
\newblock A non-ellipsoidal reachable set estimation for uncertain neural
  networks with time-varying delay.
\newblock {\em Communications in Nonlinear Science and Numerical Simulation},
  19(4):1097--1106, 2014.

\end{thebibliography}
\end{document}